\documentclass[aps,prb,10pt,twocolumn,superscriptaddress]{revtex4-2}

\usepackage{tikz,pgf}
\usepackage{amsfonts}
\usepackage{amsmath}
\usepackage{amssymb}
\usepackage{amsthm}
\usepackage{mathtools}
\usepackage{braket}
\usepackage[colorlinks=true,linkcolor=blue,citecolor=blue,urlcolor=blue]{hyperref}
\usepackage{standalone}
\usepackage{overpic}
\usepackage[right]{lineno}
\usepackage{algorithmicx,algpseudocode}
\usepackage[linesnumbered,ruled,vlined]{algorithm2e}
\usepackage{dcolumn}
\usepackage{listings}
\usepackage{subcaption}
\usepackage{caption}
\usepackage{enumitem}
\usepackage{bbm}
\usepackage{orcidlink}
\usepackage{graphicx}
\usepackage{xcolor}
\usepackage[T1]{fontenc}

\newcommand{\C}{\mathbb{C}}
\newcommand{\R}{\mathbb{R}}
\newcommand{\I}{\mathbbm{1}}

\newcommand{\ketbra}[1]{\ket{#1}\bra{#1}}
\renewcommand{\set}[1]{\left\{ #1 \right\}}

\graphicspath{}

\newtheorem{definition}{Definition}[section]
\newtheorem{theorem}{Theorem}[section]
\newtheorem{corr}{Corollary}[section]

\usetikzlibrary{arrows.meta,
                positioning,
                calc,
                shapes.geometric,
                math,
                3d}

\begin{document}

\title{A Short Note on the Generators of Controlled Quantum Gates}

\author{Richard M.~Milbradt \orcidlink{0000-0001-8630-9356}}
\email{r.milbradt@tum.de}
\affiliation{Technical University of Munich, CIT, Department of Computer Science, Boltzmannstra{\ss}e 3, 85748 Garching, Germany}
\author{Christian B.~Mendl~\orcidlink{0000-0002-6386-0230}}
\email{christian.mendl@tum.de}
\affiliation{Technical University of Munich, CIT, Department of Computer Science, Boltzmannstra{\ss}e 3, 85748 Garching, Germany}
\affiliation{Technical University of Munich, Institute for Advanced Study, Lichtenbergstra{\ss}e 2a, 85748 Garching, Germany}
\date{\today}

\begin{abstract}
We present the analytical generators for arbitrary multi-qubit controlled gates. Closed forms for the generating Hamiltonians are given for gates with both multiple control and target qubits, as well as for arbitrary control conditions. This allows us to go beyond gate-based simulations of quantum circuits and incorporate decoherence and other noise in simulations of quantum computers. We exemplify this by simulating the impact of a harmonic oscillator interacting with two qubits during the application of a controlled NOT gate.
\end{abstract}

\maketitle

We present some useful results for an alternative view of the quantum-gate picture of quantum computing. Usually, one considers the unitary operators on qubits $U$ as the building blocks of quantum circuits. Observe that for every unitary quantum gate $U$, there exists a \emph{generating Hamiltonian} or \emph{generator} $H[U]$ such that
\begin{equation}\label{eq:generator_def}
    U = e^{-iH[U]},
\end{equation}
where $H[U]$ is hermitian \cite{Reed1980}. By simulating a time evolution governed by $H[U]$ for some initial state $\ket{\psi}$ over the interval $[0,1]$, we can simulate the effect of $U$ applied to $\ket{\psi}$. Notably, this is closer to the physical implementation of quantum hardware; we can simply replace the Hamiltonian with the one that implements the given gate on the hardware. Additionally, it is easy to couple the gates to external non-qubit quantum systems to investigate realms beyond gate-based quantum computing. However, given an arbitrary $U$, it is generally not easy to find the corresponding generator. We will show that for a common class of multiqubit operators, the generators have an analytical expression.

\section{Generators of Quantum Gates}
\subsection{Single Qubit Gates}
We will start by restating some well-known generators of single-qubit quantum gates \cite{Nielsen2010, Muradian2005}. Consider the following identity
\begin{equation}\label{eq:involutory_identity}
    e^{i\theta U} = \cos(\theta) \I +i\sin(\theta) U
\end{equation}
for involutory $U$, i.e. $U^2=U$. This immediately implies
\begin{equation}
    U = e^{i\frac{\pi}{2}(\I - U)},
\end{equation}
or stated differently
\begin{equation}
    H[U]=\frac{\pi}{2}(\I - U)
\end{equation}
for involutory $U$. Involutory gates include the Pauli-Gates $X$, $Y$, and $Z$, as well as the Hadamard gate $\mathcal{H}$. The identity \eqref{eq:involutory_identity} also contains the generators of the rotation gates. By inserting the Pauli-Gates, we get
\begin{equation}
    R_U(\theta) = e^{-i \frac{\theta}{2} U} \text{ for } U \in \set{X, Y, Z}.
\end{equation}
This is by no means a coincidence, as the Pauli-Matrices are angular momentum operators and as such generate the Lie-group of rotations around the corresponding axes \cite{Zee2016}.\\

\begin{figure*}
    \begin{subfigure}{0.49\textwidth}
        \includegraphics{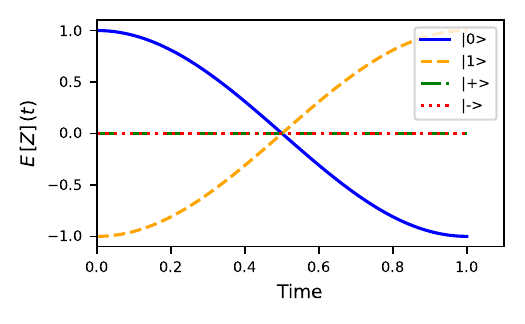}
        \caption{Pauli $X$ gate.}
        \label{fig:paulix_evolution}
    \end{subfigure}
    \begin{subfigure}{0.49\textwidth}
        \includegraphics{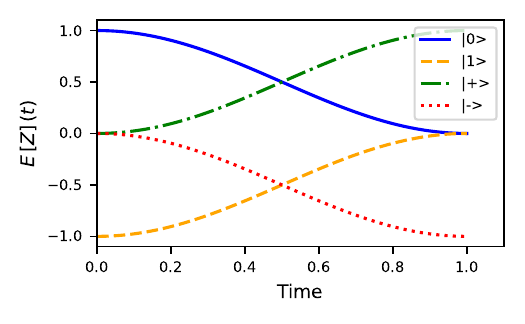}
        \caption{Hadamard gate $\mathcal{H}$.}
    \end{subfigure}
    \caption{The evolution of different initial states via single qubit gates via generator-based time-evolution.}
    \label{fig:single_gates_time_evo}
\end{figure*}

The time-evolution governed by a sample of single qubit gates, $X$ and $\mathcal{H}$, is shown in Fig.~\ref{fig:single_gates_time_evo}. We can see the expected behaviour for the different initial states. For example, consider the simulation according to the $X$ gate in Fig.~\ref{fig:paulix_evolution}. The two eigenstates of $X$, $\ket{+}$ and $\ket{-}$, do not change their expectation value with respect to $Z$, both being constant at $0$. On the other hand, the computational basis states $\ket{0}$ and $\ket{1}$ slowly change from one to the other. This is reflected by the expectation value slowly changing from $1$ to $-1$ for $\ket{0}$ and $-1$ to $1$ $\ket{1}$. In the plotted results for the Hadamard gate $\mathcal{H}$, we can clearly see the change of basis from the $X$-eigenbasis to the $Z$-eigenbasis and vice versa. However, quantum circuits consisting only of single-qubit gates do not increase a system's entanglement. Thus, they are easily simulatable using tensor network methods. To obtain less trivial circuits, we need multi-qubit gates.

\subsection{Controlled Gates}
\begin{figure*}
    \begin{subfigure}{0.49\textwidth}
        \includegraphics{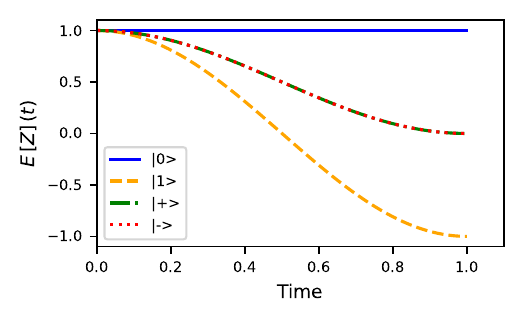}
        \caption{The expectation value $\braket{\psi | Z | \psi}$, where $Z$ is applied to one of the target qubits.}
    \end{subfigure}
    \begin{subfigure}{0.49\textwidth}
        \includegraphics{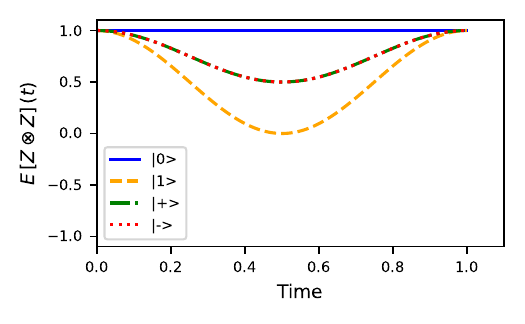}
        \caption{The expectation value $\braket{\psi | Z \otimes Z | \psi}$, where $Z\otimes Z$ is applied to both target qubits.}
    \end{subfigure}
    \caption{The expectation value of two different operators over time, when the time-evolution is governed by the $C_1XX$ generator. The states denoted in the legend are the states of the control qubit. The two target qubits are initialised as $\ket{00}$.}
    \label{fig:CXX_time_evo}
\end{figure*}

As the name suggests, multi-qubit gates act on more than one qubit. This interaction will generally increase the entanglement of the full quantum state and is required for universal quantum computing. A special class of multi-qubit gates for which we can find the generators analytically is the controlled gates.
\begin{definition}
    For a quantum gate $U \in \C^{2^L\times 2^L}$ and a control bitstring $\vec{b} \in \set{0,1}^{\times L_c}$ we define the controlled $U$ with $L_c$ control qubits as
    \begin{equation}
        C_{\vec{b}} U = \ket{\vec{b} \, }\bra{\vec{b} \, } \otimes U + \left( \I - \ket{\vec{b} \, }\bra{\vec{b} \, } \right) \otimes \I.
    \end{equation}
    Here, the left factor acts on the control qubits and the right factor acts on the target qubits.
\end{definition}

Before finding more general generators, we show the following identity:
\begin{theorem}\label{thm:single_control_generators}
    Given a gate $U$ with the generator $H[U]$, the generator of the controlled gate with respect to the bit $b \in \set{0,1}$ is
    \begin{equation}\label{eq:single_control_generator}
        H[C_b U] = \frac{1}{2} \left(\I + (-1)^b \cdot Z \right) \otimes H[U].
    \end{equation}
\end{theorem}
\begin{proof}
    As a first step, note that $H[U]$ is hermitian and as such has an eigendecomposition
    \begin{equation}
        H[U] = V \Lambda V^{\dagger}
    \end{equation}
    with $V$ unitary and $\Lambda$ a real-valued diagonal matrix. From this \eqref{eq:single_control_generator} for $b=1$ follows by some algebra
    \begin{align}
        \exp & \left( i\frac{1}{2} (\I-Z) \otimes H[U] \right) \nonumber \\
        &= \exp \left( i (\I \, \ketbra{1} \, \I) \otimes (V \Lambda V^\dagger) \right) \nonumber \\
        & = (\I\otimes V) \exp \left(\ketbra{1} \otimes \Lambda \right) (\I \otimes V^\dagger) \nonumber \\
        & = (\I\otimes V) \exp \left(0 \oplus \Lambda \right) (\I \otimes V^\dagger) \nonumber \\
        & = (\I\otimes V) \left(e^{i0} \oplus e^{i\Lambda} \right) (\I \otimes V^\dagger) \nonumber \\
        & = (\I\otimes V) \left( \ketbra{0} \otimes \I + \ketbra{1} \otimes e^{i\Lambda} \right) (\I \otimes V^\dagger) \nonumber \\
        & = \ketbra{0} \otimes \I + \ketbra{1} \otimes e^{iH[U]} \nonumber \\
        & = CU. \nonumber
    \end{align}
    The case for $b=0$ is found analogously.
\end{proof}

\begin{corr}
    For an involutory $U$, we obtain
    \begin{equation}
        H[C_bU] = \frac{\pi}{4} (\I + (-1)^b Z ) \otimes (\I - U)
    \end{equation}
\end{corr}
Thus, for the most ubiquitous two-qubit gate, we find
\begin{equation}\label{eq:cnot_generator}
    H[C_1X] = \frac{\pi}{4} (\I - Z ) \otimes (\I - X)
\end{equation}
as the generating Hamiltonian. We can easily extend our above findings to general controlled gates.

\begin{theorem}
    Given a quantum gate $U$ and a bitstring $\vec{b}$, the generator of the controlled gate is given by
    \begin{equation}\label{eq:multi_control_generator}
        H\left[C_{\vec{b}}U\right] = \left(\bigotimes_{j=1}^{n-1}\frac{1}{2}(\I + (-1)^{b_j} Z )\right) \otimes H[U].
    \end{equation}
\end{theorem}
\begin{proof}
    One can utilise the results from Thm.~\ref{thm:single_control_generators} and obtain the identity \eqref{eq:multi_control_generator} via induction.
\end{proof}
Another class of operators that we have not yet considered are controlled gates with multiple targets. That is, the target unitary consists of multiple smaller unitaries acting on pairwise disjoint supports
\begin{equation}
    U = \bigotimes_k U_k,
\end{equation}
where the set $k$ denotes the qubits on which $U_k$ acts non-trivially, i.e. the support of $U_k$. For such a gate~\cite{Neudecker1969}
\begin{equation}\label{eq:multi_target_generator}
    \bigotimes_k U_k = \bigotimes_k e^{iH_k} = e^{i\sum_k H_k},
\end{equation}
where we shortened $H\left[U_k\right]$ to $H_k$. This allows us to deduce the form of the generator of the most general kind of controlled qubit gates.
\begin{theorem}
    If $\vec{b}$ is a bitstring and $\set{U_k}_k$ quantum gates with pairwise disjoint qubit supports $k$, then
    \begin{widetext}
    \begin{equation}\label{eq:general_controlled_gate_generator}
        H\left[\left(\bigotimes_{k} C_{\vec{b}} U^{[k]}\right)\right] = H\left[C_{\vec{b}}\left(\bigotimes_{k}U^{[k]}\right)\right] = \left(\bigotimes_{j=1}^{n-1}\frac{1}{2}(\I + (-1)^{b_j} Z )\right) \otimes \left( \sum_k H\left[U^{[k]}\right] \right).
    \end{equation}
    \end{widetext}
\end{theorem}
\begin{proof}
    The first equality can be obtained in the same way as for the single-controlled phase gates in \cite{Chen2023}. Then the second equality follows directly from inserting \eqref{eq:multi_target_generator} into \eqref{eq:multi_control_generator}.
\end{proof}
In turn \eqref{eq:general_controlled_gate_generator} in the case of $U^{[k]}$ being involutory gates immediately implies
\begin{corr}
    If $\vec{b}$ is a bitstring and $\set{U_k}_k$ involutary quantum gates with pairwise disjoint qubit supports $k$, then
    \begin{widetext}
    \begin{equation}
        H\left[C_{\vec{b}}\left(\bigotimes_{k}U^{[k]}\right)\right] = \left(\bigotimes_{j=1}^{n-1}\frac{1}{2}(\I + (-1)^{b_j} Z ) \right) \otimes \left( \sum_k \frac{\pi}{2} (\I - U^{[k]}) \right).
    \end{equation}
    \end{widetext}
\end{corr}

As an example of a time-evolution governed by a multi-qubit gate, consider the application of the $C_1XX$-gate via time-evolution in Fig.~\ref{fig:CXX_time_evo}. We can see that both qubits are flipped over time, if the control is $\ket{1}$, while nothing happens if it is $\ket{0}$, as desired. Additionally, if the control bit is one of the Pauli $X$ eigenstates, the target qubits are entangled with it after the time evolution is over. Now we want to demonstrate how the generator-based simulation allows us to go beyond the quantum circuit framework.

\section{Disturbed CNOT-Gate}
\begin{figure*}
    \centering
    \begin{subfigure}{0.49\textwidth}
        \centering
        \includegraphics{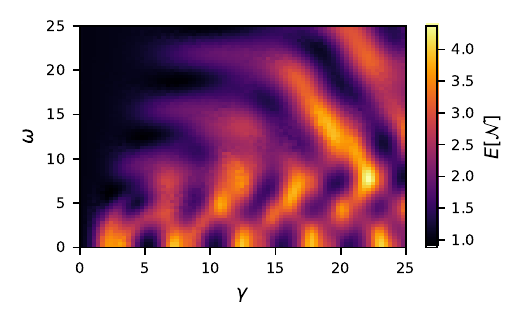}
        \caption{The number operator expectation value $E[\mathcal{N}]$.}
        \label{fig:num_op}
    \end{subfigure}
    \begin{subfigure}{0.49\textwidth}
        \centering
        \includegraphics{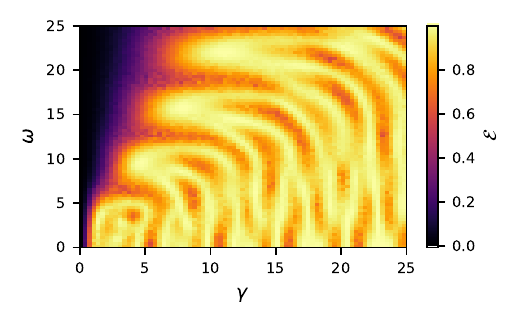}
        \caption{The error $\mathcal{E}$ compared to the perfect reference gate.}
        \label{fig:error}
    \end{subfigure}
    \caption{Results of the disturbed $C_1X$ simulation.}
\end{figure*}

Using the generator framework, we can simulate the impact of an external environment coupled to qubits. As a toy example, we consider a two-qubit system on which we run the $C_1X$ gate via the generator in \eqref{eq:cnot_generator}. The qubits are coupled to a harmonic oscillator such that the full Hamiltonian is
\begin{equation}
    H = H[C_1 X] +  \gamma X \otimes X \otimes (a + a^\dagger) + \omega \mathcal{N},
\end{equation}
where $a$ and $a^\dagger$ are the annihilation and creation operators, $\mathcal{N} = a a^\dagger$ the number operator, and $X \otimes X$ acts on the two qubits. The coupling strength $\gamma \in \R$ and characteristic frequency $\omega \in \R$ are parameters which we vary to explore the effect of the coupled oscillator. We run multiple simulations for each parameter pair, with the two-qubit initial state $\ket{\psi_0}$ chosen Haar-randomly each time. The oscillator is initially uncorrelated with the qubits and starts in the singly excited state $\ket{1}$ such that the full initial state is
\begin{equation}
    \ket{\Psi_0} = \ket{\psi_0} \otimes \ket{1}.
\end{equation}
We record the final number of excitations of the oscillator
\begin{equation}
    E [\mathcal{N}] = \braket{\Psi_f | \mathcal{N} | \Psi_f},
\end{equation}
where $\ket{\Psi_f}$ is the full system's final state. In addition, we record the error of the final state compared to the undisturbed $C_1X$ gate as
\begin{equation}
    \mathcal{E} = 1 - |\braket{\Psi_0 | C_1X | \Psi_f}|^2.
\end{equation}
The expectation value and error are recorded in Fig.~\ref{fig:num_op} and Fig.~\ref{fig:error}, respectively.

The first observation confirms intuition: stronger coupling leads to more excitations and a higher error. However, this is not a trivial relationship. Instead, for specific values of the coupling strength, the number of excitations has a global maximum. This is akin to a resonance behaviour of an oscillator at specific coupling values. While the two plots are qualitatively similar, they are only mildly correlated, with a Pearson product-moment correlation coefficient of approximately $0.653$. This is related to the specific form of the coupling term in $H$, where both excitation creation and annihilation are possible. Thus, an interaction can occur, changing the qubits' state and creating an excitation that leads to an error. Then, at a later point, the excitation is annihilated, giving the impression that no interaction happened when considering only the number operator. Accordingly, the error has a less perfect resonant behaviour and is overall more chaotic. However, we observe that the higher the characteristic frequency, the stronger the coupling must be to cause a significant error.

\section{Discussion}
This work mainly extends \cite{Muradian2005} to additional gates and is intended as a quick reference for future work on simulating quantum circuits via generators. These generators are easy to combine with tensor networks because they have an efficient matrix-product operator form \cite{cCakir2025}, which allows simulation of larger controlled gates using only local information \cite{Sander2025}. While we only used a toy model disturbance to the gate in question, it is straightforward to include the generator of a quantum gate in the system Hamiltonian for more complex open-quantum-system methods \cite{Breuer2002}, such as the hierarchical equations of motion \cite{Mangaud2023, Ke2023}. 

\begin{acknowledgments}
The research is part of the Munich Quantum Valley, which is supported by the Bavarian state government with funds from the Hightech Agenda Bayern Plus. The scripts to run the simulations and plots are available at \cite{ScriptGithub}.
\end{acknowledgments}

\bibliography{references}

\end{document}